\documentclass{svjour3}                     
\smartqed  
\usepackage{graphicx}
\usepackage{epsf}
\usepackage{psfrag}
\usepackage{booktabs}
\usepackage{enumerate}
\usepackage{ams}
\usepackage{amsxtra}
\usepackage{cite}

\usepackage{multirow}

\newtheorem{thm}{\textbf{Theorem}}
\newtheorem{cor}{\textbf{Corollary}}

\newtheorem{prop}{\textbf{Proposition}}

\newtheorem{defn}{\textbf{Definition}}
\begin{document}

\title{Prices of Anarchy, Information, and Cooperation in Differential Games\thanks{Research supported in part by grants from AFOSR and DOE.}
}


\author{Tamer~Ba\c{s}ar        \and
        Quanyan~Zhu 
}


\institute{Tamer~Ba\c{s}ar \at
              Coordinated Science Laboratory and the Department
of Electrical and Computer Engineering,  University of Illinois at Urbana-Champaign, 1308 West Main Street,
Urbana,
IL,  61801, USA. \\
              Tel.: +1 217-333-3607\\
              Fax: +1 217-265-0997\\
              \email{basar1@illinois.edu}           
           \and
           Quanyan~Zhu \at
              Coordinated Science Laboratory and the Department
of Electrical and Computer Engineering,  University of Illinois at Urbana-Champaign, 1308 West Main Street,
Urbana,
IL,  61801, USA. \\  \email{zhu31@illinois.edu}   
}

\date{Received: date / Accepted: date}

\maketitle

\begin{abstract}
The {\em price of anarchy} (PoA) has been widely  used in static games to quantify the loss of efficiency due to noncooperation.
Here, we extend this concept to  a general differential games framework. In addition, we introduce the {\em price of information} (PoI) to characterize comparative  game performances under different information structures, as well as the {\em price of cooperation} to capture the extent of benefit or loss a player accrues as a result of altruistic behavior.  We further characterize PoA and PoI for a class of scalar linear quadratic differential games under open-loop and closed-loop feedback information structures. We also  obtain some explicit bounds on these indices in a large population regime.
\end{abstract}

\section{Introduction}
\label{intro}
It is well known that the non-cooperative Nash equilibrium in  nonzero-sum games is generally inefficient \cite{DUB86}, which means that it would be possible for all players to do better in terms of attaining higher utilities or lower costs (than they would attain under Nash equilibria, even if the equilibrium is unique) through a cooperative behavior. This is true for static deterministic games, and naturally also for stochastic games as well as dynamic and differential games. In these latter of classes of games, one could bring up additional issues with regard to Nash equilibria beyond efficiency or lack thereof, such as whether an increase in information to one player (or all or a subset of the players) would be advantageous to that player (or groups of players), in terms of attaining higher utilities or lower costs, or whether acquiring more information would be undesirable for a player. In the special class of games where all players have the same utility function or cost function (that is, team problems) and what is sought is the global maximum or global minimum of these functions, the answer to such a query is clean, which is that additional information (defined as expansion of sigma fields) can never hurt. The same is true for the special class of zero-sum games. In stochastic games, or dynamic and differential games which are not team problems or zero-sum games, however, the answer is not that clean, and one could encounter quite surprising and at the outset counter-intuitive results. Perhaps the first demonstration of this was reported in \cite{Basar72} and \cite{BasHo74}, where two classes of two-player stochastic static games were considered, one a linear-quadratic-Gaussian (LQG) model and the other one a stochastic Cournot duopoly model, both of which admit unique Nash equilibria. It was shown that for the LQG model better information (on some stochastic variables) for {\em only one} player leads to lower average Nash equilibrium costs for {\em both} players, but in the duopoly model only the player whose information is improved benefits while the other one hurts (in the sense that his average Nash equilibrium cost increases). Another way of comparison would be in terms of the relative values of the average Nash equilibrium costs attained by the players, when one player has informational advantage over the other. It was again shown in \cite{Basar72}  that, in an otherwise completely symmetric game, the player who has better information attains higher cost than the other player in the LQG model (the counter-intuitive result), whereas he attains lower cost in the duopoly model (the intuitive result). Several manifestations of these conclusions can be seen also in dynamic and differential games; for example time-consistent open-loop Nash equilibrium is not necessarily inferior to the strongly time-consistent closed-loop feedback Nash equilibrium \cite{BasOls99}. 

Now coming back to  inefficiency of Nash equilibrium in a fixed nonzero-sum game, one question of interest is exploration of the extent of this inefficiency, that is how far off is a Nash equilibrium from the socially optimal solution, which is obtained as the maximum of the sum of the utilities of the players, or some convex combination of the utilities (or minimum in the case of cost functions). The notion of the {\em price of anarchy} (PoA)was introduced in  \cite{ROU04} as a quantification of this offset, as a utility ratio  between the worst possible Nash solution (among multiple Nash equilibria) and the social optimum. In a way, this index serves to quantify the loss of efficiency due to competition. It has been shown that in routing games and resource allocation games  (see, \cite{ROU04} and \cite{JMT05}),  PoA is bounded by a constant, allowing agents to achieve some level of efficiency despite being suboptimal. 

The  idea of quantifying the gap between  social optimality and  game equilibrium solutions  sparked many follow-up work in that same vein. In \cite{SS08}, {\em price of simplicity} has been introduced for a pricing game in communication networks as the ratio between the revenue collected from a flat pricing rule and the maximum possible revenue. In \cite{GJC09}, {\em price of uncertainty} has been introduced  to measure the relative payoff of an expert user of a security game under complete information  to the one under incomplete information. In \cite{ZHU08b}, {\em price of leadership} has been proposed as a measure of comparison of utilities in a power control game between Nash equilibria and Stackelberg solutions. In all of these works, primarily communication networks have been used as a backdrop application domain, be it routing, resource allocation, power control, or security. Game-theoretical methods along with Nash equilibrium have found many applications  in communication networks, with some selected recent references being \cite{AB98, Basar07, MahBas03, SAAB02, ZHU08e, JOH03, JMT05,  ABEJW06}

In this paper, we discuss several indices which quantify variations or offsets in the payoff values or costs attained under Nash equilibria in the context of differential games (DGs). We first extend the notion of {\it PoA} to DGs, which heretofore   has been primarily limited to static continuous kernel games.  We provide a characterization of {\it PoA} for a class of scalar linear-quadratic (LQ) DGs, and quantify the efficiency loss in the long run when the players behave non-cooperatively under the Nash equilibrium concept. We consider both open-loop (OL) and closed-loop (CL) information structures (ISs). We show that for the class of scalar LQ DGs with CL IS using the strongly time-consistent CL feedback Nash equilibrium, the {\it PoA} has some appealing computable upper bounds, which can  further be approximated when the number of players is sufficiently large (that is, the large population regime), whereas, under the OL IS, it is possible to obtain an expression for the {\it PoA} in closed form. 
 
As mentioned earlier, going from static to dynamic (differential) games brings in the possibility of various ISs, which add richness to the (Nash equilibrium) solution of a game. Different ISs  (generally) yield different equilibrium solutions, and hence IS is a crucial factor in the investigation of {\it PoA} in DGs. Motivated by this, we introduce another index, the {\em price of information} (PoI), which is a result of the comparison of the equilibrium utilities or costs under different ISs.  
For the class of scalar LQ DGs above, we show that the {\it PoI} between the feedback and open-loop ISs is shown to be bounded from below by ${\sqrt{2}}/{2}$ and from above by $\sqrt{2}$, again  in the large population regime.
Finally, motivated by some recent results reported in \cite{AAE10} on the level of cooperation between players in a routing game, captured by the degree of willingness of a player to place partial weight on other players' utilities in his utility function, we introduce the {\em price of cooperation} (PoC) as a measure of benefit or loss to a player on his base Nash equilibrium payoff due to cooperation.

The structure of the paper is  as follows. In Section~\ref{sec:2}, we introduce a general $N$-player DG framework with different ISs, and define in this context the indices, {\em PoA}, {\em PoI}, and {\em PoC}. In Section~\ref{sec:3}, we investigate the {\em PoA} for a class of scalar LQ feedback DGs. In Section~\ref{sec:4}, we study the LQ DGs under open-loop IS, and in Section~\ref{sec:5}, we establish bounds on the {\em PoI}. We conclude and identify future directions in Section~\ref{sec:6}.
An earlier version of some of the results in this paper can be found in the recent conference paper \cite{ZhuBas10}.

\section{General Problem Formulation}
\label{sec:2}
In this section we first introduce the general  nonzero-sum differential games framework along with the Nash equilibrium solution, and then introduce the three indices: prices of anarchy, information, and cooperation.

Let $\mathcal{N}=\{1, 2, \cdots, N\}$ be the set of players, and $[0, T\rangle$\footnote{The notation ``$\rangle$"  is introduced to capture two cases: finite horizon when $T$ is finite (in which case we have $[0, T]$), and infinite horizon when $T$ is infinite (in which case we have $[0, \infty)$).} be the time interval of interest. At each time instant $t\in [0, T\rangle$, each player, say Player~$i$,  chooses an $m_i$-dimensional control value (action) $u_i(t)$ from his set of feasible control values $U_i\subset \mathbb{R}^{m_i}$, where we also make the standard assumption that as a function of $t$ the control function $u_i(\cdot)$ is piecewise continuous on $[0, T\rangle$.  The state variable $x$ is of dimension $n$, and takes values in $ \mathbb{R}^n$; as a function of time, $t$,  we assume $x(\cdot)$ to be piecewise continuously differentiable on $[0, T\rangle$, and evolving according to the differential equation:
$$\dot{x}(t)=f(x(t),u_1(t), \cdots, u_N(t), t)\,, \;\;x(0)=x_0\,,$$
where $x_0\in\mathbb{R}$ is the initial value of the state and the system dynamics $f(\cdot): \Omega\rightarrow \mathbb{R}^n$ is defined on the set 
$$\Omega=\{(x, u_1, \cdots, u_N, t)| x\in \mathbb{R}^n, t\in[0, T\rangle, u_i\in U_i, i\in\mathcal{N}\}\,,$$ as a jointly piecewise continuous function which is also Lipschitz in $x$, and also possibly Lipschitz in the $u_i$'s, depending on whether the underlying information structure (IS) is open loop of closed loop feedback.

We will consider two different ISs: {\em Open loop} (OL),, where the controls are just functions of time, $t$ (and also of initial state $x_0$, which however is assumed to be fixed and a known parameter of the game), and {\em closed-loop state-feedback}, where the controls are allowed to be functions of current value of the state and of time, that is, for Player $i$, $u_i(t) = \gamma_i(t; x(t))$. In the latter case, $\gamma_i : [0, T\rangle \times \mathbb{R}^n \to U_i$ is known as the policy variable (strategy) of Player~$i$, which is a mapping from the set of information available to the player to his control (action) set.\footnote{One can introduce more general ISs, such as those that involve memory, but here we will restrict the discussion to only OL and CL state-feedback (SF) structures so as not to encounter {\em informational non-uniqueness} of Nash equilibria \cite{BasOls99}.} We  require  each $\gamma_i(t; \cdot)$ to be Lipschitz in $x$, in addition to being jointly piecewise continuous in its arguments, and denote the class of all such mappings by $\Gamma_i$. We further require that $f$ be Lipschitz not only in $x$ but also in $\{u_1,\ldots , u_N\}$, so that the differential equation generating the state,
$$\dot{x}(t)=f(x(t),\gamma_1(t; x(t)), \cdots, \gamma_N(t; x(t)), t)\,, \;\;x(0)=x_0\,,$$
admits a unique piecewise continuously differentiable solution for each $\gamma_i \in\Gamma_i,\; i\in \mathcal{N}$. 
Clearly, when a particular $\gamma_i$ does not depend on $x$ (such as the OL IS), then it would be captured as a special case, and hence to capture this also notationally, we  will write $\gamma_i\in\Gamma_i$ as $\gamma_i^\eta\in\Gamma_i^\eta$, where $\eta$ stands for the underlying IS (which for the discussion in this paper is either OL or CL SF).\footnote{Even though in general different players can have different ISs, we will consider here only the case when the IS in the entire DG is either OL  or CL SF. Otherwise, derivation of Nash equilibrium becomes complicated, and one has to introduce {\em small noise robustness} in order to eliminate informational non-uniqueness, even in LQ DGs \cite{Basar89}, \cite{BasOls99}. At the conceptual level, however, the analysis in this paper, and the indices introduced, equally apply to the mixed IS case.}

Each player $ i\in\mathcal{N}$ is a cost-minimizer, with the objective function for Player~$i$, as defined on the state and action spaces, is given by 
$$L_i(u)=\int_0^T F_i(x(t),u_1(t), \cdots,u_N(t), t)dt + S_i(x(T))$$
when $T<\infty$, and
$$L_i(u)=\int_0^\infty F_i(x(t), u_1(t), \cdots, u_N(t), t)dt$$
when $T=\infty$, where $u:= \{u_1,\ldots , u_N\}$. In the expressions above,
for each $i\in\mathcal{N}$, the function $F_i:\Omega \rightarrow\mathbb{R}$ is Player~$i$'s instantaneous (running) cost function, and in the first expression $S_i: \mathbb{R}^n\rightarrow \mathbb{R}$ is the terminal value function. 
Substituting  $u_i(t) = \gamma_i(t; x(t))$ in the above, we arrive at the {\.em normal} or {\em strategic} form of the DG, where now the dependence in $L_i$  is on $\gamma_i$'s instead of $u_i$'s. Let us denote this new cost function representation by $J_i$, for Player~$i$, which we write more explicitly (showing its argument) as  $J_i(\gamma^\eta)$, where $\gamma^\eta := \{\gamma_1^\eta,\ldots,\gamma_N^\eta\} \in \Gamma^\eta:= \Gamma^\eta_1 \times\cdots\times\Gamma^\eta_N$, where again this covers also the OL IS as a special case; we will occasionally drop the superscript $\eta$ when the IS is clear from context.


Let $\gamma_{-i}^\eta$ denote the collection of policies of all players except  Player~$i$, i.e.,
$\gamma_{-i}^\eta=(\gamma_1^\eta, \ldots, \gamma_{i-1}^\eta, \gamma_{i+1}^\eta, \ldots, \gamma_N^\eta)\,,$ in a game with IS $\eta$. If $\gamma_{-i}^{\eta}$ is fixed as ${\gamma_{-i}^{\eta*}}$, Player~$i$ is faced with the dynamic optimization (optimal control) problem: \footnote{We use ``$\textrm{OC}(i)$" to denote Player~$i$'s individual optimal control problem.}
\begin{align}
 (\textrm{OC}(i))\;\;& \min_{\gamma_i\in {{\Gamma_{i}^{\eta}}}}J_{i}(\gamma_i, \gamma_{-i}^{\eta*})
:=\int_0^T F_i(x,\gamma_i(\eta), {\gamma_{-i}^{\eta*}}(\eta), t)dt+ S_i(x(T))\\
\nonumber& \textrm{s.t.~} \;\;  \dot{x}(t)=f(x, \gamma_i(\eta), {{\gamma^{\eta*}_{-i}}}(\eta), t)\,,\;\; x(0)=x_0\,.
\end{align}

In the case of infinite horizon, the problem remains the same with $S_i\equiv 0$ and $T=\infty$. If we denote the solution to $\textrm{OC}(i)$ by ${\gamma_i^\eta}^*$, and carry out the optimization for each $i$, then what we have is a Nash equilibrium compatible with the IS that defines the DG. This is made precise below.
\begin{defn}\label{etaNE}
[$\eta$-Nash equilibrium] For a DG with IS $\eta$, the policy $N$-tuple $\{{\gamma_i^{\eta}}^*,\;  i\in\mathcal{N}\} =: {\gamma^{\eta}}^*$ is  an $\eta-$Nash equilibrium if, for each $i\in \mathcal{N}$, $\gamma_i^{\eta*}$ solves  the optimal control problem (OC$(i)$). Let $\Gamma^{\eta*}$ be the set  of all $\eta-$Nash equilibria, as a subset of $\Gamma^{\eta}$.
\end{defn}

Now, for the CL IS case, one has to further refine the Nash equilibrium, in order to eliminate informational non-uniqueness. Consider a family of DGs, structured the same way, but defined over the time interval $[s, T\rangle$, where $s > 0$ is the parameter that identifies different elements of the family. We say that an $\eta$-Nash equilibrium, when $\eta$ is the CL IS is {\em strongly time consistent} if its restriction to $[s, T\rangle$ is also an $\eta$-Nash equilibrium, and this being true for each $s$ and all $x(s)$. Such Nash equilibria could also be called sub-game perfect equilibria, by direct analogy with a similar concept in finite games. We will henceforth consider only strongly time consistent Nash equilibria when $\eta$ is CL, but will suppress that refinement in the development below.

Let $J^{\eta*}_i, i\in\mathcal{N}$, denote the  achieved values of the objective functions of the players under a particular $\eta-$Nash equilibrium $\gamma^{\eta*}$, and a corresponding  total cost achieved (as a convex combination of the individual costs) be given by $J_\mu^{\eta*}=\sum_{i\in\mathcal{N}}\mu_iJ_i^{\eta*}$, where $\mu_i$ is a positive weighting factor on Player~$i$'th cost, satisfying the normalization condition $\sum_{i\in\mathcal{N}}\mu_i=1$. We assume, without any loss of generality, that $J^{\eta*}_i > 0$ for all $i\in\mathcal{N}$, and hence {\em a fortiori} $J_\mu^{\eta*} > 0$.

Now as a benchmark, let us consider the case of full coordination, where the players agree on minimizing a single objective function, which is a convex combination of the individual cost functions. We may call this also a socially optimal solution.
 The corresponding underlying optimization problem is the optimal control problem: \footnote{The acronym ``COC" stands for ``Centralized Optimal Control".}
\begin{align}
\nonumber (\textrm{COC})\;\;\; & \min_{\gamma\in\Gamma} \sum_{i=1}^N\mu_i\left\{\int_0^T F_i(x(t),\gamma(\eta), t)dt + S_i(x(T))\right\}\\
\nonumber \textrm{s.t.~} & \dot{x}(t)=f(x, \gamma(\eta), t)\,,\;\; x(0)=x_0\,,
\end{align}
where the optimization could also be carried out with respect to control values, $u$, that is in an open-loop fashion, since the problem is deterministic and also is not strategic.  Hence, the optimal value of this optimal control problem is independent of the IS, which we denote by
$J_\mu^{\circ}$, and the corresponding (open-loop)  optimal control  by $u^{\circ}=[u_1^{\circ}, \ldots, u_N^{\circ}]$. Note that we necessarily have $0 < J_\mu^{\circ} \leq J_\mu^{\eta*}\,,\;$ where $J_\mu^{\eta*}$ is under any Nash equilibrium solution out of $\Gamma^{\eta*}$.

\begin{defn}
[Price of Anarchy] Consider an $N$-person DG as above and its associated optimal control problem (COC) with $J_\mu^{\circ} > 0$. The {\em price of anarchy} for the DG is\footnote{If the maximum below does not exist, then it is replaced by supremum in the definition of PoA.} 
\begin{equation}\label{GeneralPoA}
\rho_{N,\mu,T}^\eta=\max_{\gamma^{\eta*}\in \Gamma^{\eta*}} \, J_\mu^{\eta*} / {J_\mu^{\circ}}
\end{equation}
as the worst-case ratio of the total game cost to the optimum social cost.
\end{defn} 
In addition to its dependence on the cost functions,  PoA depends on the number of players in the game, the IS, the weights on individual players and the time horizon. 
Note that the PoA as defined in (\ref{GeneralPoA}) is lower-bounded by 1. 

\begin{defn}[Price of Information (PoI)] \label{GeneralPoI}
 Let $\eta_1$ and $\eta_2$ be two ISs. Consider two $N$-person DGs which differ only in terms of  their ISs,  with game $1$ having IS $\eta_1$, and game $2$ having $\eta_2$. Let the values of a particular $\mu$ convex combination of the objective functions be  ${J^{\eta_1}_\mu}^*$ and ${J^{\eta_2}_\mu}^*$, respectively, achieved under the Nash equilibria ${\gamma^{\eta_1}}^*$ and ${\gamma^{\eta_2}}^*$. The {\em price of information} between the two ISs (under cost minimization) is given by
\begin{equation}
\chi_{\eta_1}^{\eta_2}(\mu)=\max_{\gamma^{\eta_2^*}\in \Gamma^{\eta^*_2}}J_\mu^{\eta^*_2} \, /\max_{\gamma^{\eta_1^*}\in \Gamma^{\eta_1^*}}J_\mu^{\eta_1^*}.
\end{equation}
\end{defn}
The PoI compares the worst-case costs  under two different ISs for the same convex combination, and  quantifies the relative loss or gain when the DG is played under a different IS. Clearly, when $\chi_{\eta_1}^{\eta_2}(\mu)<1$, the IS $\eta_2$ is superior to its counterpart $\eta_1$ .
The connection between PoI and  PoA can be captured by 
$\;\chi_{\eta_1}^{\eta_2}(\mu)={\rho_{N,\mu, T}^{\eta_2}}\, / {\rho_{N,\mu, T}^{\eta_1}}\,.$

Before introducing the third index (price of cooperation), let us define another class of DGs, which is an intermediate case between full cooperation and full non-cooperation. Consider the case
where Player~$i$, even though his cost function is $J_i$, adopts an altruistic mode and minimizes instead a cost function that places some weight on other players' costs. Let $\lambda_i:= \{\lambda_i^j , j\in \mathcal{N}\}$ be a set of nonnegative parameters adding up to $1$,
$\sum_{j\in\mathcal{N}}\lambda^j_i=1$. Let $\tilde J_i(\gamma^\eta; \lambda_i)\,,\; i\in \mathcal{N}$ be defined by
$$ \tilde J_i(\gamma^\eta; \lambda_i):= \sum_{j\in\mathcal{N}} \lambda_i^j J_j(\gamma^\eta)\,,\;\; i\in \mathcal{N}$$
Consider the $\eta$ IS DG with cost functions $\tilde J$'s, and let $\tilde\Gamma^\eta$ be the set of all its $\eta$-Nash equilibria. For $\tilde\gamma^\eta\in \tilde\Gamma^\eta$, Player~$i$ achieves an actual cost of $J_i(\tilde\gamma^\eta)$, which may be better (lower) or worse (higher) than $J^{\eta*}_i$ defined earlier. Note that if $\lambda_i^j = \mu_i$ for all $i, j\in \mathcal{N}$, then all players have the same cost function, and every $\eta$-Nash equilibrium solution of the altruistic game is a solution to COC, assuming that person by person optimal solutions of COC are globally optimal. Hence, in this limiting case we have full cooperation. This now brings us to the third index, which is keyed to individual players.

\begin{defn}[Price of Cooperation (PoC)] \label{GeneralPoC}
Consider an N-player DG with a fixed IS $\eta$, and with a fixed set of  cooperation vectors $\lambda:= \{\lambda_i,\; i\in \mathcal{N}\}$.
Let $\tilde{J}_i,\; i\in \mathcal{N}$, and ${\tilde\Gamma}^\eta$ be as defined above, and $\Gamma^\eta$ be the set of all Nash equilibria of the original game. Then, the {\em price of cooperation} for Player~$i$ under the cooperation scheme $\lambda$ is given by
\begin{equation}\label{PoC}
\nu_i^\eta(\lambda) = \max_{\gamma\in{\tilde\Gamma}^\eta} J_i (\gamma) /  \max_{\gamma\in\Gamma^\eta} J_i (\gamma)\,.
\end{equation}
\end{defn}

As indicated earlier, if $\lambda_i = \mu$ for all $i$, where $\mu = \{\mu_i,\; i\in \mathcal{N}\}$ as in {\em PoA}, then every NE of  $\{\tilde{J}_i,\; i\in \mathcal{N}\}$ is a person-by-person optimal solution of the {\em COC} with cost function $J_\mu$, which would also be globally optimal under some appropriate convexity conditions. If $\gamma^0$ is one such solution, minimizing $J_\mu$, then the {\em PoC} is given by 
$$\nu_i^\eta(\mu) =  J_i (\gamma) /   \max_{\gamma\in\Gamma^\eta} J_i (\gamma)\,,$$
which can be viewed as the reciprocal of {\em individualized PoA}, where the latter is a measure of the loss or gain an individual player incurs on his individual cost when he (along with other players) plays the worst NE strategy as opposed to 
the globally minimizing strategy (again along with other players).

\section{Scalar LQ Feedback Differential Games}
\label{sec:3}
The analysis of the price of anarchy is complex for general DGs as there often exist more than one Nash equilibrium, which show strong dependence on the underlying IS. 
For specific game structures, however, its analysis may be tractable provided that we avoid informational non-uniqueness. One such class is scalar linear quadratic DGs with state feedback IS, which is what we focus on in this section. These games also enjoy wide applications in economics and communication networks; see, \cite{DJLS06}, \cite{AB98}.  We first state our model and recall some important relevant results on LQ feedback DGs; for details, see \cite{BasOls99}, \cite{Eng05}.

\subsection{Game Model}
\label{subsec:3.1} 
As a special case of the class of DGs considered in the previous section, consider the infinite-horizon scalar $N-$person LQ DGs, with quadratic cost function
\begin{equation}\label{cost}
L_i(u) = \int_0^\infty \left(q_i x^2(t) +r_iu_i^2(t)\right) dt,\; \;\;i\in\mathcal{N},
\end{equation}
\begin{equation}\label{system}
\dot{x}(t)=ax(t)+\sum_{i=1}^Nb_iu_i(t),\; \;\;x(0)=x_0\,,
\end{equation}
where $q_i>0$, $r_i > 0$, $x_0 \not= 0$, $b_i \not =0$ are all scalar quantities. Let  $b:=[b_1, \dots, b_N]$. We are interested in   strongly time-consistent state-feedback (SF) Nash equilibrium (NE), where  further  the NE policies are required to be stationary (that is time invariant). We will refer to such equilibria in short as {\em Feedback NE}. The following theorem provides their characterization.


\begin{thm}\label{FBNE}
[Feedback NE, \cite{BasOls99}, \cite{Eng05}]  Let $\{k_i, \; i\in\mathcal{N}\}$ solve the set of coupled algebraic Riccati equations
\begin{equation}\label{RE}
2\left(a-\sum_{j=1}^Ns_jk_j\right)k_i+q_i+s_ik_i^2=0, \; i\in\mathcal{N}
\end{equation}
satisfying the stability condition
$a-\sum_{i=1}^Ns_ik_i<0\,,$
where $s_i:=b_i^2/r_i$. Then, the $N$-tuple of policies $\gamma_i^*(x)=-\frac{b_i}{r_i}k_ix, \; i\in\mathcal{N},$ constitutes a feedback NE, with the corresponding cost for 
Player~$i$ being $J^*_i=k_ix^2_0$. Furthermore, the positively weighed total cost is $J_\mu^*=\bar{k}x^2_0$, where $\bar{k}=\sum_{i=1}^N\mu_ik_i$.\newline
If the set of coupled  algebraic Riccati equations do not admit a solution which is also stabilizing, then the DG does not have a feedback NE. \hfill$\diamond$
\end{thm}

The main challenge in computing the feedback NE solution for this DG is that equation (\ref{RE}) is a nonlinear coupled system of equations. The fact that we have a scalar problem alleviates the difficulty somewhat, since  it is possible to turn it into a linear problem through a change of variables, as outlined in \cite{Eng00a},\cite{Eng00b}.
Let $\sigma_i=s_iq_i$, $\sigma_{\max}=\max_i\sigma_i$, $p_i=s_ik_i , i=1, \ldots,N$, and 
\begin{equation}\label{lambda}
\lambda=\sum_{i=1}^Np_i-a.
\end{equation} 
Multiplying (\ref{RE}) by $s_i$, we rewrite it as
\begin{equation}\label{ModRE}
p_i^2-2\lambda p_i+\sigma_i=0, \; i=1, \ldots, N.
\end{equation}
Let $\Omega \subset \mathcal{N}$ be an index set, $\Omega_{-i}=\Omega\backslash \{i\}$, and $n_\Omega=|\Omega|$. For every $\Omega \neq \emptyset$, we have (after some manipulations)
{\small
\begin{equation}\label{eqnOmega1}
\prod_{j\in\Omega}p_j\lambda  =\frac{1}{2n_{\Omega}-1}\left\{\sum_{i\in\Omega}\sigma_i\prod_{j\in\Omega_{-i}}p_j-\sum_{i\notin\Omega}\prod_{j\in\Omega}p_jp_i+a\prod_{j\in\Omega}p_j\right\}.
\end{equation}
}
When $\Omega=\emptyset$, we define
\begin{equation}\label{eqnOmega2}
\prod_{j\in\Omega}p_j\lambda:=\lambda=\sum_{j=1}^Np_j-a.
\end{equation}
Hence, for every $\Omega$, we have an equation in the form of either (\ref{eqnOmega1}) or (\ref{eqnOmega2}). Let 
$\mathbf{p}=[1, p_1, p_2,$ $\ldots, p_N, p_1p_2, \ldots, p_1p_N, p_2p_3, \ldots, p_{N-1}p_N, \ldots, \prod_{i=1}^Np_i]^T$. We can write (\ref{eqnOmega1}) and (\ref{eqnOmega2}) into 
\begin{equation}\label{tildeMp}
\widetilde{\mathbf{M}}\mathbf{p}=\lambda \mathbf{p}.
\end{equation} 
Let $\mathbf{p}:=[1, k_1, k_2, \ldots, k_N, k_1k_2,\ldots, k_1k_N, k_2k_3, \ldots, $ $k_{N-1}k_N,  \ldots, \prod_{i=1}^Nk_i]^T$ and $\mathbf{D}=\textrm{diag}\{1$, $ s_1$, $s_2$, $\ldots$, $s_N$, $s_1s_2,$  $\ldots$, $s_1s_N$,$s_2s_3,$ \ldots, $s_{N-1}s_N, \ldots, \prod_{i=1}^Ns_i\}\,.$ Hence, we can rewrite $\mathbf{p}=\mathbf{D}\mathbf{k}$ and (\ref{tildeMp})  into
\begin{equation}\label{Mp}
\mathbf{M}\mathbf{k}=\lambda \mathbf{k}, \, \mbox{ where }\; \mathbf{M} :=\mathbf{D}^{-1}\widetilde{\mathbf{M}}\mathbf{D}\,.
\end{equation} 
Equation (\ref{Mp}) is an eigenvalue problem with each index set $\Omega$ corresponding to a row enumerated starting from the empty set. It has maximum $2^N$ distinct eigenvalues and $2^N$ eigenvectors. The vector formed by the second entry to the $N+1$-st entry of the eigenvectors yields the solution to (\ref{RE}) when the first entry of the vector is normalized to $1$ and they satisfy the stability condition of Theorem~\ref{FBNE}.  This leads to:


\begin{thm}\label{FBNEComp}
[Feedback NE Computation, \cite{Eng05}] Suppose $\mathbf{M}$ is a nondefective matrix with distinct eigenvalues. Let $(\lambda, \mathbf{k})$ be an eigenvalue-eigenvector pair such that $\lambda\in\mathbb{R}_+$ and $\lambda>\sigma_{\max}$. Then, a feedback NE $\gamma_i^*(x)=-\frac{b_i}{r_i}k_i\,x, \; i\in\mathcal{N}\,,$  is yielded by
$k^*=\mathbf{1}^T\mathbf{k}$ provided that the resulting solution is stabilizing, where $\mathbf{1}=[0, 1, \ldots, 1, 0, \ldots, 0]^T$ is a vector whose $2$nd to $N+1$-st entries are 1's.
\end{thm}

\begin{thm}\label{Uniq}
[Uniqueness of Feedback NE] Let $\bar{p}:=\sum_{j\in\mathcal{N}}p_j, p_{-i}:=\sum_{j\in\mathcal{N}, j\neq i} p_j$. There exists a unique feedback NE for the LQ DG described by (\ref{cost}) and (\ref{system}) under either one of the following two conditions:\newline
(i)  $N$ is sufficiently large such that $p_{-i}>a, \forall i$, or (ii)
 $a=0$.\newline
Furthermore, the solutions to the coupled algebraic Riccati equations that characterize the feedback NE are of the following forms under the corresponding conditions above:
\begin{enumerate}[(s-i)]
\item $p_i=(\bar{p}-a)-\sqrt{(\bar{p}-a)^2-\sigma_i}~;$
\item $p_i=\bar{p}-\sqrt{\bar{p}^2-\sigma_i},\;$ where
\end{enumerate}
\begin{equation}\label{FPeqnP}
\bar{p}-a=\frac{1}{N-1}\left(\sum_{i=1}^N\sqrt{(\bar{p}-a)^2-\sigma_i}+a\right).
\end{equation}
Moreover, the stability condition $a-\sum_{i=1}^Ns_ik_i<0\,$ is satisfied, and hence the FB NE is stabilizing.
\end{thm}
\begin{proof}
 From (\ref{ModRE}), we obtain
\begin{equation}
p_i^2+2(p_{-i}-a)p_i-\sigma_i=0,
\end{equation}
which admits the solutions:
\begin{equation}\label{pipni}
p_i=(a-p_{-i})\pm\sqrt{(a-p_{-i})^2+\sigma_i}.
\end{equation}
Since we need $p_i>0$, we retain the one with $``+"$ sign. By rearranging the positive solution of (\ref{pipni}), we arrive at
\begin{equation}\label{pi1}
(\bar{p}-a)^2=(p_{-i}-a)^2+\sigma_i\,,
\end{equation}
and, therefore, in terms of $\bar{p}$, we have
\begin{equation}\label{pi2}
p_i=(\bar{p}-a)\pm\sqrt{(\bar{p}-a)^2-\sigma_i}.
\end{equation}
Under condition (i), we have $p_i-\bar{p}+a<0$, hence we obtain the unique solution (s-i). Under scenario (ii), (\ref{pi2}) reduces to $p_i=\bar{p}\pm\sqrt{\bar{p}^2-\sigma_i}.$ Since, $p_i<\bar{p}$, we again obtain the unique solution (s-ii).

By summing over (\ref{pi2}), we have a fixed point equation (\ref{FPeqnP}). Let 
$$\bar{P}(\bar{p}):=\frac{1}{N-1}\left(\sum_{i=1}^N\sqrt{(\bar{p}-a)^2-\sigma_i}+a\right)-(\bar{p}-a)\,.$$ 
Its derivative is given by
$$\frac{d\bar{P}}{d\bar{p}}=-1+\frac{\bar{p}-a}{N-1}\left(\sum_{i=1}^N\frac{1}{\sqrt{(a-\bar{p})^2-\sigma_i}}\right).$$
Since $\sigma_i\geq0$ and $\bar{p}-a>0$, it follows that 
\begin{eqnarray}
\frac{d\bar{P}}{d\bar{p}}&\geq&-1+\frac{\bar{p}-a}{N-1}\left(\frac{N}{(\bar{p}-a)}\right)\\
&=&\frac{1}{N-1}>0, \textrm{~for~}N\geq2.
\end{eqnarray}
This says  that $\bar{P}$ is a monotonically increasing function, and hence the solution to $\bar{P}=0$ is unique. Hence, under (i) or (ii), there exists a unique feedback NE.

The fact that the solution is stabilizing follows directly from (\ref{RE}), where the first term has to be negative because the second and third terms are positive.
\end{proof}

\subsection{Team Model}
\label{subsec:3.2}
When  players form a team to achieve an optimal social objective, a specific total cost is minimized. Let $\bar{q}_\mu=\sum_{i=1}^N\mu_iq_i$,  $\overline{R}_\mu=\textrm{diag}\{\mu_1r_1, \ldots, \mu_Nr_N\}$, and consider
\begin{align}
\nonumber (\textrm{FOC})& \;\;\;\; \min_{u(t)} \int_0^\infty \left(\bar{q}_\mu x^2(t)+u^T(t)\overline{R}_\mu u(t)\right)dt \\
\nonumber \textrm{s.t.~} &\;\;\; \dot{x}(t)=ax(t)+\sum_{i=1}^Nb_iu_i(t)\,,\;\; x(0) = x_0 \not=0\,.
\end{align}

The solution to this optimal control problem is standard, and is given below for future reference (where we suppress the dependence of $\bar{q}$ and $\overline{R}$ on $\mu$).

\begin{thm}\label{OC}
[Centralized Optimization]  The optimal control problem (FOC) admits a unique feedback solution which is further stabilizing. The optimal policies are 
\begin{equation}\label{OCeqn}
\gamma^\circ_i (x)=-\frac{b_i}{\mu_ir_i}\hat{k}_\mu\,x\,, \quad \hat{k}_\mu: =\frac{a+\sqrt{a^2+\bar{q}\bar{b}}}{\bar{b}}\,,
\end{equation} 
with $\bar{b}:=\sum_{i=1}^N (b_i^2 / \mu_ir_i)$, and   minimum cost is  $J^\circ_\mu=\hat{k}_\mu x_0^2$.\newline
The optimal control can also be expressed in open-loop form, as:
$$u^\circ_i=-\frac{b_i}{\mu_ir_i}\hat{k}_\mu\Phi(t,0)x_0,$$
where $\Phi(t,0)$ is the unique solution to
$$\dot{\Phi}(t,0)=\left(a-\sum_{i=1}^N\frac{b_i^2}{\mu_ir_i}\hat{k}_\mu\right)\Phi(t,0), ~~\Phi(0,0)=1.$$
\end{thm}

\subsection{Price of Anarchy (PoA)}
\label{subsec:3.3}
Here, we provide a closed-form expression for the PoA in the feedback LQ DG, where we make the natural assumption that $x_0 \not= 0$, as otherwise the costs are all zero. 
\begin{thm}\label{PoAthm}
The PoA of the LQ feedback  DG described by (\ref{cost}) and (\ref{system}) is characterized by the following:
\begin{enumerate}[(i)]
\item Given a weight vector $\mu$,  the PoA $\rho_{\mu}$ is equal to
\begin{equation}\label{statement1}
\rho_{\mu}^{FB}=\max_{\mathbf{k}\in\mathcal{K}} \,\,[ \,{\boldsymbol{\mu}^T\mathbf{k}}\, ]\, / {\hat{k}}\,,
\end{equation}
where $\boldsymbol{\mu}=[0, \mu^T, 0, \ldots, 0]^T$ and $\mathcal{K}$ is the set of all eigenvectors of the matrix $\mathbf{M}$.
\item Suppose $\mu_i=\bar{\mu}_i:={s_i}\, / {\sum_{j=1}^Ns_j}, i\in\mathcal{N}$. Then,
$$\rho^{FB}_{\bar{\mu}}\leq  [\, {\varrho({\mathbf{M}})+a}\, ]\, / {\sum_{i=1}^Ns_i\hat{k}}\,,$$ where $\varrho(\mathbf{M})$ is the spectral radius of $\mathbf{M}$.
\item Let $\mu^s_{\max}=\max_{i\in\mathcal{N}}\mu_i/s_i$. Given a weight vector $\mu$ that satisfies $\sum_{i=1}^N\mu_i=1$, the PoA is bounded  by
\begin{equation}\label{PoAInequality}
\rho^{FB}_{\mu}\leq {\mu^s_{\max}(\varrho(\mathbf{M})+a)}\,/ {\hat{k}}.
\end{equation}
\end{enumerate}
\end{thm}
\begin{proof}
The proof is a direct application of the results in Theorem~\ref{FBNE} and Theorem~\ref{OC}.  PoA is the worst-case ratio of the game cost under feedback NE to the optimum social cost as defined in (\ref{GeneralPoA}). Under the feedback IS, an LQ DG has $$\rho^{FB}_{\mu}=\max_{k^*}\frac{\sum_{i=1}^N\mu_ik^*_i(x_0)^2}{\hat{k}(x^0)^2}=\max_{\mathbf{k}\in\mathcal{K}}\frac{\mu^T\mathbf{k}}{\hat{k}}\,.$$ This leads to statement (i). The price of anarchy under $\bar{\mu}$ is 
\begin{eqnarray}
\nonumber \rho^{FB}_{\bar{\mu}}&=&\max_{k}\frac{\sum_{i=1}^N\bar{\mu}_ik_i}{\hat{k}}=\max_{k}\frac{s_ik_i}{\sum_{i=1}^N{s_i}\hat{k}}\\  &=&\max_{\lambda}\frac{\lambda+a}{\sum_{i=1}^N{s_i}\hat{k}}.\end{eqnarray}
The last equality is due to (\ref{lambda}). Hence, by taking the largest eigenvalue, we obtain (ii). The equality is achieved when $\varrho(\mathbf{M})$ is an eigenvalue in the eigenvalue-eigenvector pair that yields the equilibrium from Theorem \ref{FBNEComp}.  For an arbitrarily picked $\mu$, (\ref{statement1}) yields 
\begin{eqnarray}
\nonumber \rho^{FB}_{\bar{\mu}}&=&\max_k\frac{\sum_{i=1}^N\frac{\mu_i}{s_i}s_ik_i}{\hat{k}} \leq \max_k\frac{u^s_{\max}\sum_{i=1}^Ns_ik_i}{\hat{k}}\\ 
&=&\max_{\lambda}\frac{u_{\max}^s(\lambda+a)}{\hat{k}}\leq\frac{u_{\max}^s(\varrho(\mathbf{M})+a)}{\hat{k}}.
\end{eqnarray}
Using (\ref{lambda}) and taking the worst case, we obtain statement (iii). Since $$\max_{i\in \mathcal{N}}\frac{\bar{\mu}_i}{s_i}=\frac{1}{\sum_{j=1}^N{s_j}}\,,$$ the  last inequality is achieved when $\mu=\bar{\mu}$.
\end{proof}
\bigskip

The next corollary further characterizes the  bound on PoA.
\begin{cor}
The following follow from Theorem \ref{PoAthm}:
\begin{enumerate}[(i)]
\item Given a $\mu$ and $a\neq 0$, PoA is bounded  above by
\begin{equation}
\rho^{FB}_{{\mu}}\leq\left(1+\frac{1}{2a}(N+\sigma_{\max}-1)\right)s^\bullet,
\end{equation}
where $\sigma_{\max}=\max_{i\in\mathcal{N}}\sigma_i$, and $$s^\bullet :=\sum_{i=1}^N\frac{s_i}{\min_{j\in\mathcal{N}}s_j}\,.$$
The upper-bound is independent of  $\mu$.

\item If $a=0$, PoA is bounded  above by
\begin{equation}\label{COR2}
\rho^{FB}_{{\mu}}\leq\frac{\mu^s_{\max}}{\sqrt{\bar{q}}\sqrt{\mu^s_{\min}}}\sqrt{N}(N+\sigma_{\max}-1),
\end{equation}
where $\mu^s_{\min}=\min_{i\in\mathcal{N}}\mu_i/s_i$.
\end{enumerate}
\end{cor}

\begin{proof}
The matrices $\mathbf{M}=[m_{ij}]$ and $\widetilde{\mathbf{M}}=[\tilde{m}_{ij}], i,j =1,\ldots, 2^N,$ share the same set of eigenvalues. From Gersgorin theorem, we can obtain {\small
$$\varrho(\widetilde{\mathbf{M}})\leq \min\left\{\max_i\sum_{j=1}^{2^N}|\tilde{m}_{ij}|, \max_j \sum_{i=1}^{2^N}|\tilde{m}_{ij}|\right\}\leq \max_i\sum_{j=1}^{2^N}|\tilde{m}_{ij}|.$$}
From (\ref{eqnOmega1}) and (\ref{eqnOmega2}), the absolute row sum $RS_k, k=1,\ldots, 2^N$, can easily be  evaluated by letting $p_i=1$:
\begin{eqnarray}\label{eqnOmega3}
\nonumber RS_k
= [{a+\sum_{i\in\Omega}\sigma_i+(N-n_\Omega)}]\, / \, [{2n_{\Omega}-1}], 
\end{eqnarray}
where $k$ is the row index corresponding to the set $\Omega$.
When $\Omega=\emptyset$, we let $RS_1= N+a$.
From (\ref{PoAInequality}), $$\rho^{FB}_{\mu}\leq\,\,[\,{\varrho(\mathbf{M})+a}\,]/ ({\hat{k}/\mu_{\max}^s}).$$ The numerator is upper-bounded by (skipping some steps):
\begin{eqnarray}\label{Nineq}
\nonumber \varrho(\mathbf{M})+a
\nonumber&\leq& \max\left\{\max_{1\leq n_\Omega\leq N}\frac{(2a+\sigma_{\max}-1)n_\Omega+N}{2n_\Omega-1}, 2a+N-1 \right\}\\
\nonumber&\leq& \max\left\{{2a+N+\sigma_{\max}-1}, 2a+N-1\right\}\\
&\leq& 2a+N+\sigma_{\max}-1.
\end{eqnarray}
The 
second inequality holds because the quantity  
$$\frac{(2a+\sigma_{\max}-1)n_\Omega+N}{2n_\Omega-1}$$ increases with $n_\Omega$.
The denominator has a lower bound:
\begin{eqnarray}\label{Dineq}
\nonumber\frac{2a}{\bar{b}\mu^s_{\max}}&\geq&\frac{2a}{\sum_{i=1}^N\left(\frac{\max_{i\in\mathcal{N}}\mu_i/s_i}{\mu_i}\right)\frac{b_i^2}{r_i}}\\
&\geq&\frac{2a}{\sum_{i=1}^N\frac{s_i}{\min_{i\in\mathcal{N}}s_i}}=\frac{2a}{s^\bullet}\,.
\end{eqnarray}
 The last inequality is due to $\max_i \mu_i/s_i\leq\max_i\mu_i\max_i\frac{1}{s_i}$.
Combining (\ref{Nineq}) and (\ref{Dineq}), we have, for $a\not= 0$,
$$\rho^{FB}_{{\mu}}\leq\left(1+\frac{1}{2a}(N+\sigma_{\max}-1)\right)s^\bullet$$
When $a=0$, $$\hat{k}=\sqrt{\bar{q}/\bar{b}}=\sqrt{\frac{\bar{q}}{\sum_{i=1}^N\frac{s_i}{\mu_i}}}\geq\frac{\sqrt{\bar{q}\mu^s_{\min}}}{\sqrt{N}}\,.$$ Using this together with (\ref{Nineq}), we  arrive at the inequality (\ref{COR2}).
\end{proof}
\bigskip

The upper bound on price of anarchy in the preceding corollary provides a worst case of efficiency loss.

The next result studies the large population game and its proof relies on the Taylor series expansion of the square-root term in (\ref{pi2}).
\begin{thm}\label{ApproxThm}
Suppose the number of players in the LQ DG is sufficiently large so that 
$$ \mbox{(C-i) } p_{-i}>a, \forall i\in\mathcal{N}\,,\; \mbox{(C-ii) } a \ll N\,,\; \mbox{(C-iii) } \sigma_{\max}\ll \bar{\sigma}\,,$$
where $\bar{\sigma}=\sum_{i=1}^N\sigma_i$. Then, the following quantities can be approximated as given:
$$\mbox{(i) } p_i\sim\frac{\sigma_i}{\sqrt{2\bar{\sigma}}}\,,\;\; \mbox{(ii) }  u_i\sim-\frac{\sigma_i}{b_i\sqrt{2\bar{\sigma}}}x\,,\;\; $$ $$\mbox{(iii) }  J^*\sim \frac{\bar{q}}{\sqrt{2\bar{\sigma}}}(x_0)^2\,,\;\;
\mbox{(iv) } J^*\sim \frac{\bar{q}}{\sqrt{2\bar{\sigma}}}(x_0)^2\,,\;\; $$ $$\mbox{(v) } \rho^{FB}_\mu\sim\frac{\bar{q}}{\hat{k}\sqrt{2\bar{\sigma}}}\,,\; \mbox{and for } a=0, \, \rho^{FB}_{\mu}\sim\sqrt{\frac{\bar{q}\bar{b}}{2\bar{\sigma}}}\,.$$
\end{thm}
\begin{proof}
By Taylor series expansion, (\ref{pi2}) can be written as
\begin{eqnarray}\label{TE}
\nonumber p_i&=&(\bar{p}-a)\left[1-\sqrt{1-\frac{\sigma_i}{(\bar{p}-a)^2}}\right]\\
&=&\frac{\sigma_i}{2(\bar{p}-a)}\left[1+O\left(\frac{\sigma_i}{(\bar{p}-a)^2}\right)\right],
\end{eqnarray}
where $O(\cdot)$ is a function such that $\lim_{x\rightarrow0}O(x)=0$. In a similar way, (\ref{FPeqnP}) can be rewritten as (skipping some steps):
\begin{align}
\nonumber \bar{p}-a =&\frac{\bar{p}-a}{N-1}\left(\sum_{i=1}^N\sqrt{1-\frac{\sigma_i}{(\bar{p}-a)^2}}+a\right)\\
=&\frac{\bar{p}-a}{N-1}\left[\frac{N\bar{\sigma}}{2(\bar{p}-a)^2}\left(1+O\left(\frac{\sigma_{\max}}{2(\bar{p}-a)^2}\right)\right)+a\right].
\end{align}
Hence, we obtain for large $N$
\begin{eqnarray}\label{pi3}
\bar{p}-a
&=&\sqrt{\frac{\bar{\sigma}}{2}}\left[1+O\left(\frac{\sigma_{\max}}{2(\bar{p}-a)^2}\right)\right]
\end{eqnarray}
Note that $\bar{p}-a>0$ due to the stability condition. Let $\bar{\sigma}=\sum_{i=1}^N\sigma_i$, as before.
Let a solution of (\ref{pi3}) be $\bar{p}=\sqrt{\bar{\sigma}/2}+a$, i.e.,
\begin{eqnarray}\label{pi4}
\bar{p}-a=\sqrt{\frac{\bar{\sigma}}{2}}\left[1+O\left(\frac{\sigma_{\max}}{\bar{\sigma}}\right)\right].
\end{eqnarray}
(\ref{pi4}) is consistent provided that $\sigma_{\max}\ll \sigma$ and $a\ll N$. Since, by Theorem \ref{Uniq}, the solution is unique under (C-i),  $\bar{p}$ can indeed be approximated by $\bar{p}\sim a+\sqrt{\bar{\sigma}/2}$, which leads to $p_i\sim\frac{\sigma_i}{\sqrt{2\bar{\sigma}}}$ from (\ref{TE}). Hence,   (ii)-(v) follow.
\end{proof}

\section{Open-Loop LQ Differential Games}
\label{sec:4}
In this section, we go back to the DGs described by (\ref{cost}) and (\ref{system}), but with open-loop information. Each player knows only the value of the initial state of the system. Since the cost runs from zero to infinity, we  are interested in controls that yield finite costs. Accordingly, we restrict the controls of the players to belong to the set
$$\mathcal{U}^{OL}(x_0)=\{u\in \mathcal{L}_2 [0,\infty) \mid J_i(x_0,u) < \infty,\; \forall i\in  \mathcal{N}\}\,,$$
where $\mathcal{L}_2 [0,\infty)$ is the space of square-integrable functions on $[0, \infty)$.

\begin{thm}\label{OLthm}
[Open-Loop NE, \cite{BasOls99}, \cite{Eng05}]  Consider the $N-$person LQ DG in (\ref{cost}) and (\ref{system}), and assume that there exists a unique solution $\xi^\star$ to the set of
equations
\begin{equation}\label{OLRE}
0=2a\xi_i+q_i-\xi_i\left(\sum_{j=1}^Ns_j\xi_j\right),
\end{equation}
such that $a-\sum_{j=1}^Ns_j\xi_j^\star<0$, where $s_i := b_i^2 / r_i$. Then, the game admits a unique open-loop Nash equilibrium for every initial state, given by
\begin{equation}\label{OLNE}
u_i^\star(t)=-\frac{b_i}{r_i}\xi_i^\star \exp \left[\left(a-\sum_{j=1}^Ns_j\xi_j^\star\right)t\right]x_0\,.
\end{equation}
The optimal cost to player $i$ using $u_i^\star$ is $\, J_i^\star=k_i^\star x_0,\;$ where $k_i^\star$ is the unique solution to
\begin{equation}\label{OLLE}
2\left(a-\sum_{j=1}^Ns_j\xi_j^\star\right)k_i+q_i+s_i(\xi_i^\star)^2=0.
\end{equation}
\end{thm}

The quantities in Theorem~\ref{OLthm} can  be made more explicit as we discuss below. By a slight abuse of notation, let $p_i:=s_i\xi_i$ as in the state-feedback information case.  
Multiplying (\ref{OLRE}) and (\ref{OLLE}) by $s_i$, we obtain
$\;0=2ap_i+\sigma_i-p_i\bar{p}\,,\;$
and $0=2s_ik_i(a-\bar{p})+\sigma_i+p_i^2,$
where $\bar{p}=\sum_{i=1}^Np_i$. Hence we can solve for $p_i, k_i$, and obtain
\begin{equation}\label{OLpi}
p_i={\sigma_i}\, / \, ({\bar{p}-2a})\end{equation}
\begin{equation}\label{OLkappa}
k_i={\sigma_i+p_i^2}\, /\, ({2s_i(\bar{p}-a)}).
\end{equation}
To obtain $\bar{p}$, we sum  (\ref{OLpi}) over $i$ and arrive at the quadratic equation
$\bar{p}=
\frac{\bar{\sigma}}{\bar{p}-2a}.$
Thus,
\begin{equation}\label{barp}
\bar{p}=\sqrt{a^2+\bar{\sigma}}+a\,,
\end{equation}
where we have retained only the positive solution of the quadratic equation for obvious reasons.  It should be pointed out that since the relevant $\bar{p}$ is unique, we have a unique  open-loop NE.
Using (\ref{barp}), we can determine the expression for $\xi_i^\star$ (and thus the OL NE strategies of the players \ref{OLNE}),  as
\begin{equation}\label{xistar}
\xi_i^\star=\frac{q_i}{\sqrt{a^2+\bar{\sigma}}-a}.
\end{equation}
Note that these are necessarily stabilizing, that is $a-\sum_{j=1}^Ns_j\xi_j^\star<0$, in view of (\ref{OLLE}).
Now using (\ref{barp}) and (\ref{OLpi}) in (\ref{OLkappa}), we arrive at the  closed-form expression for $k_i^\star$:
\begin{equation}
k_i^\star=\frac{1}{\sqrt{a^2+\bar{\sigma}}}\left(\frac{q_i}{2}+\frac{\sigma_iq_i}{2(\sqrt{a^2+\bar{\sigma}}-a)^2}\right).
\end{equation}
When $a=0$, $k_i^\star$ is reduced to \begin{equation}\label{kappaSoln}
k^\star_i=\frac{1}{\sqrt{\bar{\sigma}}}\left(\frac{q_i}{2}+\frac{\sigma_iq_i}{2\bar{\sigma}}\right).
\end{equation}
Given  weighting $\mu$, the open-loop NE yields a total cost of $$J_\mu^\star=\sum_{i=1}^N\mu_iJ_i^\star=\sum_{i=1}^N\mu_i k_i^\star (x_0)^2 =: k_\mu^\star (x_0)^2\,.$$
Since the open-loop NE solution is unique, the PoA under open loop IS can thus be easily found to be: 
\begin{equation}\label{OLPoA}
\rho_{\mu}^{OL}={k_\mu^\star}\,/\,{\hat{k}_\mu}\,.
\end{equation}
We now capture all this in the corollary below.

\begin{cor}\label{OLNE unique}
The OL LQ DG of Theorem~\ref{OLthm} admits a unique OL NE given by (\ref{OLNE}) and (\ref{xistar}), which is also stabilizing. Furthermore, the OL PoA is given by (\ref{OLPoA}).
\end{cor}


\section{Price of Information (PoI)}
\label{sec:5}
In the previous sections, we have introduced PoA as a measure of efficiency in going from cooperative to noncooperative framework, and obtained expressions for it for FB and OL LQ DGs . Here, we study the price of information (PoI) as a measure of efficiency with respect to the ISs for again the LQ DG. Following  Definition \ref{GeneralPoI},  PoI between open-loop and feedback ISs is defined by
\begin{equation}\label{PoI}
\chi^{OL}_{FB}={\max_{k^\star}J^{OL\star}}\, / \, {\max_{k^*}J^{FB*}}\,,
\end{equation}
which can also be expressed  in terms of the {\em PoA}s under the two ISs:
$$
\chi_{FB}^{OL}={\rho_{\mu}^{OL}}\, / \,{\rho_{\mu}^{FB}}\,.
$$ Using Theorem~\ref{PoAthm}, we can obtain
a bound on PoI: $$\chi^{OL}_{FB}\geq \frac{k^\star}{\mu_{\max}^s(\varrho(\mathbf{M})+a)}\,.$$ The following theorem further characterizes the PoI in a special case.

\begin{thm}\label{PoIthm}
Suppose $a=0$, and the number of players is large so that $N$ satisfies (C-i), (C-ii), and (C-iii). Then, the PoI is bounded from  above  and  below by two constants:
\begin{equation}\label{PoIbound}
{\sqrt{2}}/{2}\leq\chi_{FB}^{OL}\leq\sqrt{2}.
\end{equation}
\end{thm}
\smallskip

\begin{proof}
Under conditions  (C-i), (C-ii), and (C-iii), we have a unique feedback NE that can be approximated as in statement (iv) of Theorem~\ref{ApproxThm}. Hence, from (\ref{barp}) we obtain
\begin{eqnarray} \nonumber
\nonumber \chi_{FB}^{OL}&=&\frac{J^{OL\star}}{J^{FB*}}=\frac{\sqrt{2}}{2}\left(1+\frac{\sum_{i=1}^N\mu_iq_i\sigma_i}{\bar{q}\bar{\sigma}}\right)\\
\nonumber &=&\frac{\sqrt{2}}{2}\left(1+\frac{\sum_{i=1}^N\mu_iq_i\sigma_i}{\sum_{i=1}^N\mu_iq_i\sum_{i=1}^N\sigma_i}\right) \leq \sqrt{2}\,,
\end{eqnarray}
where the last inequality is obtained by noting that $$\sum_{i=1}^N\mu_iq_i\sigma_i\geq\sum_{i=1}^N\mu_iq_i\sum_{i=1}^N\sigma_i\,.$$ The lower bound can be achieved by noting that $\sigma_i, q_i, \mu_i$ are all nonnegative.
\end{proof}
\bigskip

Theorem~\ref{PoIthm} is useful  in the design of games via access control or pricing mechanisms. Let $\bar{\chi}\in(\frac{\sqrt{2}}{2},\sqrt{2}]$ be some target PoI to achieve so that $\chi_{FB}^{OL}\leq\bar{\chi}$. For example, when $\bar{\chi}=1$, it means the game needs to be designed so that the open-loop NE yields no larger cost than the feedback NE. Hence, a necessary condition to meet such a design criterion is:
\begin{equation}\label{PoIDesignCond} \frac{\sum_{i\in\mathcal{N}} \mu_iq_i\sigma_i}{\bar{q}\bar{\sigma}}\leq \sqrt{2}\chi_{FB}^{OL}-1.\end{equation}
An access control is to admit a set $\mathcal{N}$ of players so that (\ref{PoIDesignCond}) is satisfied when all the system and player parameters are given. When set $\mathcal{N}$ is fixed and not adjustable, we may use ``pricing" mechanisms to control the parameters $r_i$ or $q_i$, which reflect the unit ``price" of penalty on the control effort and the state, respectively. In the following corollary, we capture the special case of homogeneous players.

\begin{cor}\label{cor:PoI}
Suppose the LQ DG satisfies the conditions in Theorem \ref{PoIthm}. In addition, let the players  be symmetric so that $\sigma_i=\sigma, p_i=p,\forall i\in\mathcal{N}$. When 
$N\geq 3$,
the open-loop IS yields better total optimal cost; otherwise the FB information does better. In addition, as $N\rightarrow \infty$, $\lim_{N\rightarrow\infty}\chi_{FB}^{OL}=\frac{\sqrt{2}}{2}$ at the rate of $O\left(\frac{1}{N}\right)$.
\end{cor}

\begin{proof}
The proof directly follows from Theorem~\ref{PoIthm}. The price of information under the additional assumptions becomes $\chi_{FB}^{OL}=\frac{1}{\sqrt{2}}\left(1+\frac{1}{N}\right)$. It is independent of the parameters of the players and approaches $\frac{\sqrt{2}}{2}$ as $N\rightarrow \infty$. By letting $\chi_{FB}^{OL} \leq 1$, we obtain
$\; N\geq\, {1}\,/\,({\sqrt{2}-1})\,.$
Hence, since $N$ is an integer, the open-loop NE does better than the feedback NE when there are $3$ or more players.
\end{proof}
\bigskip

Theorem~\ref{PoIthm} and Corollary~\ref{cor:PoI} have implications in the design of games via access control when open loop is the preferred mode of play.

\section{Applications and Illustrations}
\label{sec:6}
In this section, we apply the results obtained heretofore to two classes of application scenarios in flow control.
\subsection{Multiuser Rate-Based Flow Control}
\label{subsec:6.1}
We adopt here the communication systems model described in \cite{AB98}, where the players are the {\em users} or {\em sources}, and the action (control) variables are the {\em flows} into the network. If a link receives more total flow than what it can accommodate (measured by its capacity), then packets queue up. Having long queues is not desirable, because it leads to delays in transmission. We call such links which are congested {\em bottleneck links}, and formulate the game around one such link. Let $q_l(t)$ denote the queue length at such a bottleneck link and let $s(t)$ denote the total effective service rate available at that link. Assume that each user is assigned a fixed proportion of the available bandwidth; more specifically, the traffic of source $i, i=1, 2,\ldots, N$, has an allotted bandwidth of $w_is(t)$, where $w_i$'s are positive parameters which add up to $1$. We assume that the users have perfect measurement of $s(t)$, but occasionally exceed or fall short of the bandwidth allotted to them due to fluctuations. Hence, if  $d_i(t)$ denotes the rate of source $i$ at time $t$, we can introduce  $u_i(t):=d_i(t)-w_is_r(t)$ as the control (action) variable of the source. Then, queue build-up is governed by the differential equation
\begin{equation}\label{QueueModel}
\dot{q_l}(t)=\sum_{i=1}^Nu_i(t)\,,
\end{equation}
where we assume that queue is relatively tightly controlled so that end effect constraints (starvation and exceeding an upper limit) do not become active.
The goal is  to ensure that the bottleneck queue size  stays around some desired level $\bar{q}_l$, and  good tracking between input and output rates is achieved. Toward that end, we consider the shifted variable $x(t):= q_l(t)-\bar{q}_l$, which satisfies the following differential equation which is the  shifted version of (\ref{QueueModel}):
\begin{equation}\label{MQM}
\dot{x}(t)=\sum_{i=1}^Nu_i\, \;\;x(0)=x_0\,.
\end{equation}
We now consider a noncooperative scenario in which each source determines a linear feedback policy (or an open-loop policy) to minimize its own individual cost function 
\begin{equation}\label{GP}
L_i(u)=\int_0^\infty\left(|x(t)|^2+|u_i(t)|^2\right) dt,
\end{equation}
which is consistent with the overall goal of keeping $x$ and $u_i$'s small.
We can also consider a related team problem in which sources minimize cooperatively a common cost under the same information structure (where as we know actually the IS does not make a difference in this case):
\begin{equation}\label{TP}
L(u)=\int_0^\infty\left(N|x(t)|^2+\sum_{i=1}^N|u_i(t)|^2\right) dt.
\end{equation}

This is now within the framework of LQ DGs studied earlier, with the correspondences being $a=0, x_0=1, \sigma_i= s_i= q_i=r_i =b_i=1$ in  (\ref{cost}) and (\ref{system}).  To obtain some numerical results, let us take $x_0=1$.

In the case of the 2-person LQ feedback game, the $M$ matrix introduced earlier becomes
$$\mathbf{M}_2=
\left[
\begin{array}{cccc}
0 & 1& 1& 0\\
    1& 0& 0& -1\\
    1 &0& 0& -1\\
    0 &1/3 &1/3 &0
\end{array}
\right]
$$
and  if $N=3$, we have
$$\mathbf{M}_3=
\left[
\begin{array}{cccccccc}
0 &1 &1& 1&  0&  0 & 0&  0  \\
    1& 0 &0& 0& -1 &-1 & 0&  0\\
    1 &0 &0& 0& -1&  0&  -1& 0\\
    1 &0& 0& 0&  0 &-1&  -1& 0\\
    0& 1/3 &1/3 &0& 0& 0& 0 &-1/3\\
    0& 1/3& 0 &1/3 &0& 0& 0& -1/3\\
    0& 0 & 1/3& 1/3 &0& 0& 0 &-1/3\\
    0 &0 &0 &0& 1/5& 1/5 &1/5& 0\end{array}
\right].
$$
The positive eigenvalue of $\mathbf{M}_2$ is $\lambda_2=1.1547$ and the corresponding vector is 
$\mathbf{p}_2=\mathbf{k}_2=[1.0000,    0.5774,    0.5774, 0.3333]^T$. The sum of the optimal costs under equal weights is $J^*_2=0.5774$ while the optimal common cost is $J^\circ_2=0.5$, yielding the price of anarchy value $\rho^{FB}_{\mu,2}=1.1547$. For the case with 3 players, the eigenvector is found to be $\mathbf{p}_3=\mathbf{k}_3=[1.0000,    0.4472,    0.4472,    0.4472,    0.2000,    0.2000,    0.2000,    0.0894]^T
$ corresponding to $\lambda_3=1.3416$.  Again under equal weights, the total NE  cost is $J^*_3=0.4472$ and the minimum social cost is $J^\circ_3=0.3333$. Hence, the price of anarchy is given by $\rho^{{FB}}_{\mu,3}=1.3416$. When the number of players becomes large, $\rho^{FB}_{\mu}\sim\sqrt{\frac{N}{2}}$ from Theorem \ref{ApproxThm}.


In the case of open-loop flow control, we obtain $k^\star_i=\frac{1}{\sqrt{N}}\left(\frac{1}{2}+\frac{1}{2N}\right)$ and total NE cost as $J^\star_N=k^\star$. In the 2-user game, $J^\star_2=0.5303$ yielding the price of information $\chi_{FB}^{OL}=0.9184$. The open-loop NE thus yields $8.16\%$ less cost in comparison to the closed-loop FB one. In a 3-user game, $J^\star_3=0.3849$, leading to a price of information value of $\chi_{FB}^{OL}=0.8607$, which yields a $13.93\%$ more cost for the FB IS  case. We also note that as the number of players increases, the open-loop IS yields a cost approaching $0$, i.e., $\lim_{N\rightarrow\infty}J_N^\star=0$, while in the feedback case, even though it still converges to $0$, the rate is slower: $J^*\sim\frac{1}{\sqrt{2N}}\rightarrow0$. We observe that $\chi_{FB}^{OL}$ goes to $\frac{\sqrt{2}}{2}$ at a rate of $\frac{1}{N}$ as $N$ gets large, i.e., 
$$\lim_{N\rightarrow\infty}\chi_{FB}^{OL}=\frac{\sqrt{2}}{2}+\frac{1}{2N}\rightarrow \frac{\sqrt{2}}{2}.$$ It is also noted that open-loop NEs always yield less equilibrium costs even though they require less information.

Due to the symmetry of players in the flow control problem, we can obtain exact closed-form solutions to the equilibrium costs using (\ref{pi2}) and (\ref{FPeqnP}) without approximation. It is not hard to show that under equal weights, $$J^*_{\mbox{OL}}=k_i=\frac{1}{\sqrt{2N-1}}\,, \;\; J^\star_{\mbox{FB}}=\frac{1}{\sqrt{N}}\left(\frac{1}{2}+\frac{1}{2N}\right)\, \;\;\mbox{and}\;\; J^\circ=\frac{1}{N}\,.$$ In Figure \ref{PoIOLvsFB}, we show the price of information under open-loop and feedback information structures, and in Figure \ref{PoAOLvsFB}, we show the corresponding prices of anarchy. By exact calculation, we find when $N=4$, the open-loop NE cost to be $J^\star_4=\frac{3}{8}=0.3125$, which catches up with and becomes better than the feedback NE cost: $J^*_4=\frac{1}{\sqrt{7}}=0.378$.This is consistent with  our earlier observation based on large population approximation.

We observe in Figure~\ref{PoIOLvsFB} that the NE costs are the same at $N=1$ (as they should be), and as $N$ increases, both open-loop and feedback NE costs decrease. As $N$ becomes large, both costs  approach $0$. This happens because the queue length is fixed.  When the number of players goes to infinity, the contribution from each user is negligible. Moreover, the state $x(t)$ can be driven to zero very fast as the amount of total control effort increases with the number of players. The cost incurred from the transient behavior of $x(t)$ then goes to zero. In addition, for $N\geq 2$, open-loop NE yields better costs.  The price of information $\chi_{FB}^{OL}$ is always below $1$ but maintains its level above $\frac{\sqrt{2}}{2}$. In Figure~\ref{PoAOLvsFB}, the price of anarchy starts at $1$ when $N=1$ and increases as the number of players grows. The cost under the feedback NE grows faster than the one under open-loop NE.
\begin{figure}
\begin{center}
  \includegraphics[scale=0.5]{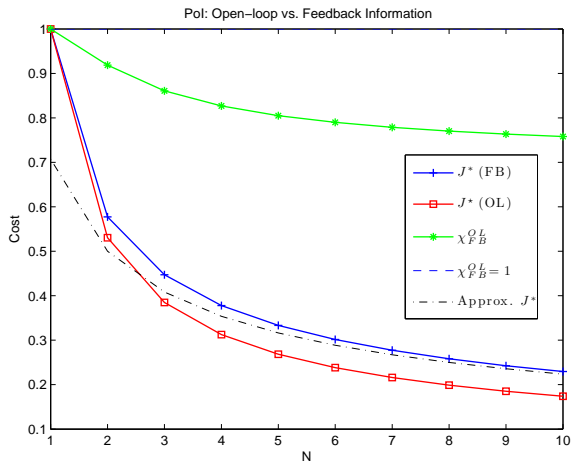}
  \caption{Price of Information}
  \label{PoIOLvsFB}
\end{center}
\end{figure}

\begin{figure}
\begin{center}
  \includegraphics[scale=0.5]{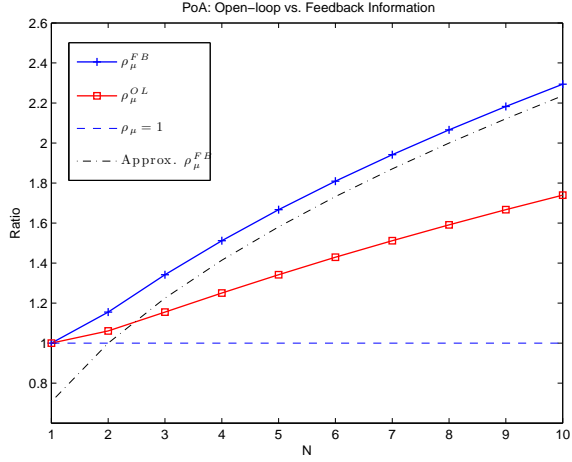}
  \caption{Price of Anarchy}
  \label{PoAOLvsFB}
\end{center}
\end{figure}

\subsection{Normalized Flow Control Dynamics}
\label{subsec:6.2}
In this section, we investigate a general flow control dynamics, which differs from (\ref{MQM}) by inclusion of a population-dependent normalization factor $f(N)$, where $f(\cdot)$ is an increasing function of $N$:
\begin{equation}\label{MQM2}
\dot{x}(t)=\frac{1}{f(N)}\sum_{i=1}^Nu_i\, , \;\;\;x(0)=1\,.
\end{equation}
The introduction of a normalization factor is to adjust the queue length proportionally when the number of users increases.
\begin{prop}
The prices of anarchy $\rho_{\mu}^{OL}, \rho_{\mu}^{FB}$, and  the price of information $\chi_{FB}^{OL}$ are independent of the normalization factor $f(N)$, as  summarized in Table \ref{FactorSummary}.
\begin{table*}[ht]
\caption{Various indices for normalized flow control game}
\begin{center}
\begin{tabular}{|c|c|c|c|c|c|}
\toprule
 $J^*$ (FB) & $J^\circ$ (TP) & $J^\star$ (OL)   & $\rho_{\mu}^{FB}$   & $\rho_{\mu}^{OL}$  & $\chi_{FB}^{OL}$  \\
\midrule
$\frac{f(N)}{\sqrt{2N-1}}$&$\frac{f(N)}{N}$&$\frac{f(N)}{\sqrt{N}}\left(\frac{1}{2}+\frac{1}{2N}\right)$&$\frac{N}{\sqrt{2N-1}}$&$\sqrt{N}\left(\frac{N+1}{2N}\right)$&$\sqrt{2-\frac{1}{N}}\left(\frac{1}{2}+\frac{1}{N}\right)$\\
\bottomrule
\end{tabular}
\end{center}
\label{FactorSummary}
\end{table*}%

\end{prop}
\begin{proof}
Using (\ref{pi2}) and (\ref{FPeqnP}), we obtain $p_i$ for a given $N$ as follows:
\begin{eqnarray}
\nonumber \bar{p}&=&\frac{N}{f(N)}\frac{1}{\sqrt{2N-1}},\\
\nonumber p_i&=&\frac{1}{f(N)\sqrt{2N-1}},\\
\nonumber k_i&=&\frac{p_i}{s_i}=\frac{f(N)}{\sqrt{2N-1}},\\
\nonumber J^*&=&\sum_{i=1}^N\frac{1}{N}k_ix_0^2=k_i.
\end{eqnarray}
The team problem yields an optimal cost of
\begin{equation}
J^\circ=\sqrt{\frac{\bar{q}}{\bar{b}}}=\frac{f(N)}{N}.
\end{equation}
Hence, the price of anarchy $\rho_{\mu}^{FB}$ under the state-feedback information structure is independent of $f(N)$, and  is given by
\begin{equation}
\rho_{\mu}^{FB}=\frac{N}{\sqrt{2N-1}}
\end{equation}
The open-loop price of anarchy is also  independent of the factor $f(N)$. 
Since $J^\star=\frac{f(N)}{\sqrt{N}}\left(\frac{1}{2}+\frac{1}{N}\right)$, it is given by
\begin{equation}
\rho_{\mu}^{OL}=\sqrt{N}\left(\frac{N+1}{2N}\right).
\end{equation}
The price of information is also independent of $f(N)$, and given by
\begin{equation}
\chi_{FB}^{OL}=\sqrt{2-\frac{1}{N}}\left(\frac{1}{2}+\frac{1}{N}\right).
\end{equation}
\end{proof}
\bigskip

As a case study, we let $f(N)=\frac{1}{N}$. Then, $b_i=\frac{1}{N}$, $s_i=\sigma_i=\frac{1}{N^2}$, for all $i\in\mathcal{N}$. When the population is large, we have $J^*\sim\sqrt{\frac{N}{2}}$ and $J^\star=\sqrt{N}\left(\frac{1}{2}+\frac{1}{N}\right)$.
The price of anarchy remains $\rho\sim\sqrt{\frac{N}{2}}$. The price of information remains $\chi_{FB}^{OL}=\frac{\sqrt{2}}{2}+\frac{\sqrt{2}}{2N}\rightarrow\frac{\sqrt{2}}{2}$ as $N\rightarrow \infty$. It can be shown that $\chi_{FB}^{OL}$ does not change with the factor $f(N)$. In Figures \ref{NormalizedPoIOLvsFB} and \ref{NormalizedPoAOLvsFB}, we show the prices based on the exact closed form solution obtained in the same fashion as in the previous section based on  (\ref{pi2}) and (\ref{FPeqnP}). We observe that the open-loop NE always outperforms the feedback equilibrium. It should be pointed out that (i) in Figure \ref{NormalizedPoIOLvsFB}, the open-loop and feedback costs increase with the number of users. This is due to the introduction of normalization factor into the system dynamics. We allocate the queue length as an increasing function of the number of users;  (ii) Figures \ref{NormalizedPoAOLvsFB} and \ref{PoAOLvsFB} are identical due to the above proposition.

If we set $f(N)=\sqrt{N}$, we have the open-loop and feedback optimal costs approach $\frac{1}{2}$ and $\frac{\sqrt{2}}{2}$ respectively, as $N\rightarrow \infty$. Figure~\ref{SQRTNNormalizedPoIOLvsFB} demonstrates that result.

\begin{figure}
\begin{center}
  \includegraphics[scale=0.5]{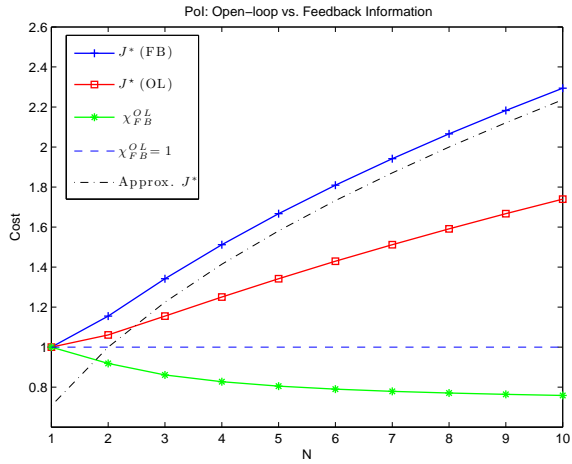}
  \caption{Price of Information  in the Normalized System, $f(N)={N}$}
  \label{NormalizedPoIOLvsFB}
\end{center}
\end{figure}

\begin{figure}
\begin{center}
  \includegraphics[scale=0.5]{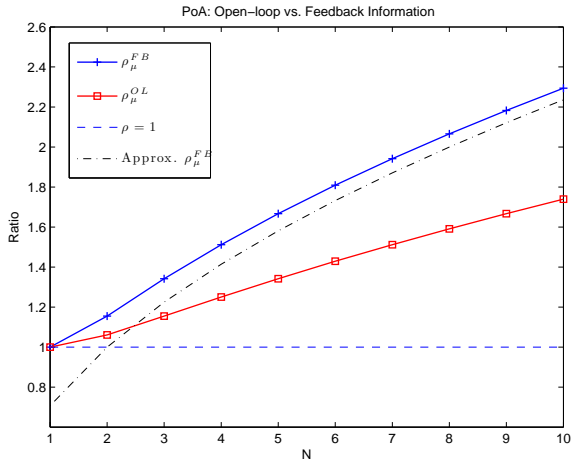}
  \caption{Price of Anarchy in the Normalized System, $f(N)={N}$}
  \label{NormalizedPoAOLvsFB}
\end{center}
\end{figure}

\begin{figure}
\begin{center}
  \includegraphics[scale=0.5]{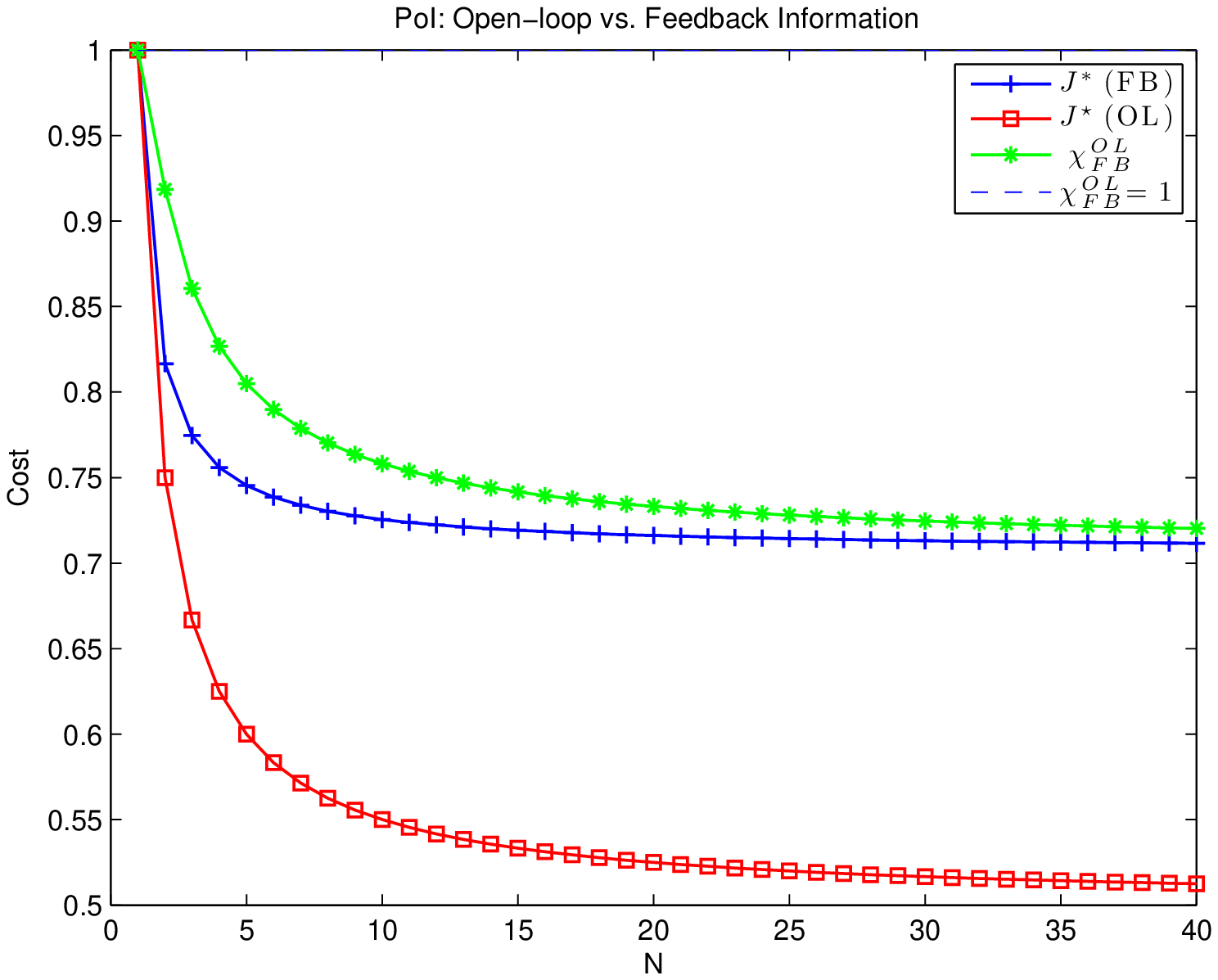}
  \caption{Price of Information  in the Normalized System, $f(N)=\sqrt{N}$}
  \label{SQRTNNormalizedPoIOLvsFB}
\end{center}
\end{figure}

A summary of the results with $f(N)=1$ and $f(N)=\frac{1}{N}$ under large population approximation is provided in Table~\ref{summary}.
\begin{table*}[t]
\caption{Indices under two  normalization factors using the large population approximation}
\begin{center}
\begin{tabular}{|c|c|c|c|c|c|c|}
\toprule
  $f(N)$ & $J^*$ (FB) & $J^\circ$ (TP) & $J^\star$ (OL)   & $\rho_{\mu}^{FB}$ & $\rho_{\mu}^{OL}$    & $\chi_{FB}^{OL}$  \\
\midrule
  $1$ &$\frac{1}{\sqrt{2N}}$&$\frac{1}{N}$&$\frac{1}{\sqrt{N}}\left(\frac{1}{2}+\frac{1}{2N}\right)$&$\sqrt{\frac{N}{2}}$& $\sqrt{N}\left(\frac{1}{2}+\frac{1}{2N}\right)$&$\frac{\sqrt{2}}{2}+\frac{\sqrt{2}}{2N}$\\
  \hline
  $\frac{1}{N}$ &$\sqrt{\frac{N}{2}}$&$1$&$\sqrt{N}\left(\frac{1}{2}+\frac{1}{2N}\right)$&$\sqrt{\frac{N}{2}}$&$\sqrt{N}\left(\frac{1}{2}+\frac{1}{2N}\right)$ &$\frac{\sqrt{2}}{2}+\frac{\sqrt{2}}{2N}$\\
\bottomrule
\end{tabular}
\end{center}
\label{summary}
\end{table*}%
\vspace{-1mm}
\section{Conclusion}

In this paper, we have introduced the notions of  price of anarchy,  price of information, and price of cooperation for nonzero-sum differential games,  have studied the first two extensively  for a  class of scalar linear-quadratic differential games, and  have obtained bounds and approximations on them, with computable bounds available in  the large population regime. Future promising work is to extend these results to non-scalar differential games as well as to obtain their counterparts for the price of cooperation.  Also computing these indices for specific models from communication networks and economics would be a fruitful area of research. 





\end{document}